\documentclass{cta-author}

\usepackage[colorlinks=true,allcolors=blue,bookmarks=false,hypertexnames=true]{hyperref} 
\usepackage{dsfont}
\usepackage{mathptmx}
\newcommand\T{{ \mathrm{\scriptstyle T} }}
\DeclareRobustCommand{\rchi}{{\mathpalette\irchi\relax}}
\newcommand{\irchi}[2]{\raisebox{\depth}{\footnotesize{$#1\chi$}}}
\newcommand\x{\rchi}

\newtheorem{Proposition}{Proposition}{}
{}
{}

\usepackage{algorithm}
\usepackage{algorithmic}

\begin{document}

\supertitle{}

\title{\large Bi-Smoothed Functional Independent Component Analysis for EEG Artifact Removal}

\author{\au{Marc Vidal$^{1,2\corr}$}, \au{Mattia Rosso$^{1}$}, \au{Ana M. Aguilera$^{2\corr}$}}

\address{\add{1}{Institute of Psychoacoustics and Electronic Music (IPEM), Ghent University, Belgium}
\add{2}{Department of Statistics and O.R. and IMAG, University of Granada, Spain}
\email{aaguiler@ugr.es (A.M.A); marc.vidalbadia@ugent.be (M.V)}}

\begin{abstract}
Motivated by mapping adverse artifactual events caused by body movements in electroencephalographic (EEG) signals, we present a functional independent component analysis based on the spectral decomposition of the kurtosis operator of a smoothed principal component expansion. A discrete roughness penalty is introduced in the orthonormality constraint of the covariance eigenfunctions in order to obtain the smoothed basis for the proposed independent component model. To select the tuning parameters, a cross-validation method that incorporates shrinkage is used to enhance the performance on functional representations with large basis dimension. This method provides an estimation strategy to determine the penalty parameter and the optimal number of components. Our independent component approach is applied to real EEG data to estimate genuine brain potentials from a contaminated signal. As a result, it is possible to control high-frequency remnants of neural origin overlapping artifactual sources to optimize their removal from the signal. An \texttt{R} package implementing our methods is available at CRAN. 
\end{abstract}

\maketitle

\section{Introduction}\label{sec1}

In the field of neurophysiology, electroencephalography (EEG) represents one of the few techniques providing a direct measure of bioelectrical brain activity, as oscillations in excitability of populations of cortical pyramidal cells \cite{Wang10} contribute to variations in the electrical potentials over the scalp. Oscillations are characterized by dominant intrinsic rhythms conventionally grouped into frequency bands, which are by now validated as markers of several neurocognitive phenomena \cite{Buzsaki06}. However, despite the temporal resolution achievable with its high sampling rate, EEG is a technique that suffers from low signal-to-noise ratio. This is mainly due to the fact that the layers of tissue dividing the electrodes from the cortex act as a natural filter attenuating genuine brain activity, resulting in a combination of cortical and artifactual sources in the EEG signal. In addition, the dominant brain-related spectral features often overlap with artifactual activity in higher frequency bands \cite{Castellanos06}, and particularly at lower frequencies most of the variance in the signal is explained by physiological sources outside the brain. For these reasons, analyzing EEG signals can ultimately be viewed as solving a source-separation problem with the goal of estimating brain potentials of interest. 

Blind source separation techniques such as independent component analysis (ICA) are commonly used to address artifact detection and correction of EEG signals. The term ICA encompasses a broad scope of algorithms and theoretical rudiments aligned to the assumption of independence of the latent sources in the data. From the statistical perspective, it could be regarded as a refinement of principal component analysis (PCA) that goes beyond the variance patterns of the data, introducing high-order measures such as kurtosis or negentropy to get more interpretable outcomes. This way, the data can be approximately represented in terms of a small set of independent variables, while in the PCA reduction, these variables are only assumed to be uncorrelated. An overview of statistical methodologies for ICA is provided in \cite{Nordhausen18}.  A comprehensive monograph of the subject can be found in \cite{Hyvarinen01}.

The use of sampling units in form of functions that evolve on a continuum, rather than through vectors of measurements, has been popularized over the last two decades to solve a broad class of problems. Functional data analysis provides a natural generalization for a wide variety  of statistical techniques that take advantage of the complete functional form of data by including relevant information related to smoothness and derivability (see \cite{Ramsay05,Hsing15,Wang16} for a systematic review of the topic).  The extension of ICA to functional data has, however, not yet received the attention nor the prolific developments of other reduction techniques in this framework, such as functional principal component analysis (FPCA).  A first attempt to develop an extension of the classic multivariate ICA  model was investigated in \citep{Mehta2009} by exploiting the functional principal component decomposition. Functional ICA techniques were also implemented in \cite{Pena14}, who defined the kurtosis operator of a standardized sample in an approximation to a separable infinite-dimensional Hilbert space. Under this setting, the 
kurtosis eigenfunctions
are expected to be rougher as the space 
does not contain functions that are pointwise convergent. Their approach focuses on the classification properties of the kurtosis operator, whose decomposition is assumed to have a similar form to the Fisher discriminant function. More recently, \cite{Li16,Virta20} developed a functional ICA model using an estimation procedure stemmed from the finite Karhunen-Lo\`{e}ve (K-L) expansion \cite[pp. 37]{Ash75}, which is a less rough space since its orthogonal expansion is optimal in the least-squared error sense. We extend this model setup 
endowing the space with a new geometrical structure given by a Sobolev inner product in order to control the roughness of the K-L functions.

The use of functional data in brain imaging analysis has gained notoriety in the last years, despite the complexity and computational cost arisen in its treatment. Data acquired from an electroencephalogram might elicit a wide variety of functional data methods, going from the estimation of smoothed sample curves to more advanced reduction and forecast techniques. See, for example, \cite{Xiao16,Hasenstab17,Nie18,Pokora18,Scheffler2018}. Current research is mainly focused on functional principal component approaches for modelling data free of artifactual sources. However, the efficiency of functional ICA techniques used in stages where data is contaminated by physiological artifacts remains, to the best of our knowledge, untested. In contrast, this problem has been extensively addressed in the multivariate environment; \cite{Uriguen15} compares the state-of-the-art methods for artifact removal. 

In this paper, a methodology based on piecewise polynomial smoothing (B-splines) is developed to disentangle the overlap between neural activity and artifactual sources.  Because of the transient nature and the complex morphology of EEG data, B-splines provide a good alternative to represent the non-sinusoidal behaviour of the neural oscillatory phenomena due to its well-behaved local smoothing. The goal is to use the proposed smoothed functional ICA to get more accurate brain estimates by subtracting artifacts free of noise. While for a strictly different kind of data, wavelet-based approaches or hybrid settings combining wavelet with ICA have been demonstrated to perform well at denoising common artifacts (see, e.g., \cite{Castellanos06, Akhtar12,Mammone14,Baja20}). 
By contrast, and despite the obvious differences between both data kinds, our independent component estimation is based on a penalized spline (P-spline) approach \citep{Eilers96,Durban02} that has a lower computational cost and is mathematically simpler. P-splines have been successfully applied for dimension reduction \cite{Aguilera13} as well as for the estimation of different functional regression models  \cite{Agui2013,Agui2016,Agui2017,Agui2020}.

Nevertheless, what characterizes our method is that the decomposition is naturally regulated by the principal component eigendirections and optimized by penalized estimators. Contrarily, in using the wavelet approaches, this is decided on the basis of the frequency band features of the data or the components. For this reason, the proposed functional ICA can be conceived as a bi-smoothed estimation procedure. The end-user will finally appreciate how artifact extraction can be fine-tuned by regulating a single smoothing parameter, making it intuitive to improve the results through a visual inspection of the independent component scores.

The paper is organized as follows.  We introduce our model in Section \hyperref[sec2]{2} and develop the smoothed FICA decomposition using basis expansion representations of functional data in Section \hyperref[sec3]{3}. A method for selecting the tuning parameters is discussed in  Section \hyperref[sec4]{4}. To test the effectiveness of our model in recovering brain signals, Section \hyperref[sec5]{5} provides a simulation using real EEG data on single trial designs containing stereotyped artifacts. 
Section \hyperref[sec6]{6} shows how our smoothed FICA works in the context of event-related potentials designs. Finally, we conclude with a brief discussion in Section \hyperref[sec7]{7}. The presented P-spline smoothed FICA is implemented in the \texttt{R} package  \textit{pfica} \citep{Vidal20}.

\section{Smoothed functional independent component analysis}
\label{sec2}
\subsection{Preliminaries}
Let $y_{i}=(y_{i1},\ldots,y_{im_{i}})^\T$ be a  signal of $i,(i=1,\ldots,n)$ components digitized at the sampling points $t_{ik},\left(k=1,\ldots,m_{i}\right)$. Consider that the sample data is observed with error, so that it can be modeled as 
\begin{equation}\label{modelofuncional}
y_{ik}=x_{i}(t_{ik})+\varepsilon_{ik},
\end{equation}
where $x_{i}$ is the $i$th functional trajectory of the signal and $\varepsilon_{ik}$ mutually independent measurement errors with zero means. The sample functions $x_{1},\dots,x_{n}$ are assumed to be realizations of independent and identically distributed copies of a random functional variable  $X$ in $L^{2}(T)$, the separable Hilbert space of square integrable functions from $T$ to $\mathbb{R}$, endowed with the usual inner product $\langle f,g\rangle=\int_{T}f(t)g(t)\mathrm{d}t,$ and the induced norm $\|f\|=\langle f,f\rangle^{1/2}$. Thorough the text, $X$ is assumed to have zero mean and finite fourth moments, which implies that  higher order operators are well defined. 

For $s,t\in T$,  the sample covariance operator $\mathcal{C}_x$ is an integral operator with kernel $c(s,t)=n^{-1}\sum_{i=1}^{n}x_{i}(s)x_{i}(t)$ admitting the Mercer decomposition 
\begin{equation*}
c(s,t)=\sum_{j=1}^\infty \eta_{j}\gamma_{j}(s)\gamma_{j}(t),
\end{equation*}
where $\{ \eta_{j},\gamma_{j}\} _{j}$ is a positive sequence of eigenvalues in descending order and their associated orthonormal eigenfunctions. The functions $x_{i}(t)$ can be approximately represented by a truncated series of the K-L expansion 
\begin{equation} \label{K-L}    x_{i}^{q}(t) = \sum_{j=1}^{q} z_{ij}\gamma_{j}(t), \end{equation} 
where $z_{ij}=\langle x_i,\gamma_{j}\rangle $ are zero mean random variables with var$(z_{j})=\eta_{j}$ and cov$(z_{j},z_{j'})=0$ for $j\neq j'.$ These variables are referred to as the principal components scores and are uncorrelated generalized linear combinations of the functional variable with maximum variance. Moreover, if the $q$ term in (\ref{K-L}) is optimally selected, the mean squared error is minimized, providing the best linear approximation to the original data \citep{Ghanem91} (pp. 21). A functional Varimax rotation has been recently introduced to improve the interpretation of the most explicative principal component scores \citep{Acal2020}. 

\subsection{Functional ICA of a smoothed principal component expansion}
\label{ss22}
The notion of independent components of a random vector cannot be immediately extended to the case of Hilbert-valued random elements (functional data) due to the fact that a probability density function is not generally defined in this context \cite{Delaigle10}. In the sequel, we consider the definition of independence introduced in \cite{Gutch12}, which establishes that a functional random variable has independent components if the coordinates obtained after projecting on to a given orthonormal basis are independent variables. Then, the aim of functional independent component analysis (FICA) is to find a linear operator $\Gamma,$ such that for a truncated orthonormal basis $ \phi_j \ (j=1,\dots, q)$ in $L^2(T)$, the variables $\langle \Gamma X, \phi_j \rangle$ are mutually independent.
By considering $X$ prompted by a Gaussian process, a functional principal component analysis (FPCA) would suffice to obtain the independent components \cite[pp. 40]{Ash75} . However,  as functional data is not inherently of this kind, it is assumed that if $X$ has a finite-dimensional representation, then it can be transformed by the operator $\Gamma$ to achieve the goals of the model. This begs the question of the basis choice for $X$, whereupon the results markedly depend.

In this paper, the sample $x_i$ is approximated by a smoothed functional PCA representation obtained by introducing an orthonormality constraint with respect to the weighted Sobolev inner product
\begin{equation} \label{eqip}
    \left\langle f,g\right\rangle {}_{\lambda} =  \left\langle f,g\right\rangle +
    \lambda \left\langle Rf,Rg\right\rangle,
\end{equation}
where $R$ is an operator with the action $Rf(t)=\mathrm{d}^{2}f(t)/\mathrm{d}t^{2}, f\in \mathrm{dom}(R)$ that measures the roughness of the curves, and $\lambda$ is a non-negative penalty parameter. 
Notice that, when $\lambda=0$, (\ref{eqip}) is simplified to the usual inner product, meaning that $x_i$ can be uniquely represented by the K-L basis, i.e. the eigenfunctions of $\mathcal{C}_x.$ To estimate the smoothed principal components, Silverman \cite{Silverman96} proposed the following variance maximization problem
\begin{equation} \label{eqsv}
\gamma_{\lambda,j}=\mathrm{arg max}\frac{\mathrm{var}(\left\langle \gamma,x\right\rangle)}{||\gamma||^{2}+\lambda\left\langle R\gamma,R\gamma\right\rangle }=\mathrm{max}\frac{\left\langle \gamma,\mathcal{C}_x\gamma\right\rangle }{||\gamma||_{\lambda}^{2}},
\end{equation}
subject to the constraint $\langle\gamma, \gamma_{\lambda,k}\rangle_{\lambda}=0$ for all $k<j,$  where $\gamma$ is a function assumed in a closed linear subspace of $L^2$ with square-integrable second derivatives. We emphasize that, the  problem of finding $\gamma_{\lambda,j}$ depends on the sample size $n$ and the selection of the penalization parameter $\lambda$. In \cite{Qi11}, the authors established the existence of the solutions of the optimization problem (\ref{eqsv}) for any $\lambda\geq0$. Silverman \cite{Silverman96} proved the consistency of the estimators as $n\rightarrow \infty$ and $\lambda\rightarrow 0$. Generalized consistency and asymptotic distributions of the estimators have been derived in \cite{Permantha17}, using expansions of the perturbed eigensystem of a sample smoothed covariance operator.

The functions $\{\gamma_{\lambda,j} \}$ form  a complete orthonormal system in the subspace endowed by $\left\langle \cdot,\cdot\right\rangle _{\lambda}$, making this basis non-compatible for our independent component model in $L^{2}(T).$ 
However, \cite{Ocana99} generalized Silverman's  method providing the following equivalents functional PCA.

\begin{Proposition}
Given a sample $\{ x_i \}$ of a functional variable with trajectories in $L^{2}(T),$ there exists a positive definite  operator $S^2$ such that the following PCA decompositions are equivalent:   
\begin{enumerate}
	\item The FPCA of $S^2(x_i)$ with respect to $\langle \cdot,\cdot\rangle_\lambda$,
	$S^2(x_i) = \sum_j z_{ij} \gamma_{\lambda, j}.$
	\item The FPCA of $S(x_i)$ with respect to $\langle \cdot,\cdot\rangle$,
	 $S(x_i) = \sum_j z_{ij} S^{-1}(\gamma_{\lambda, j}).$
	\item The FPCA of $X$ with respect to $\langle \cdot,\cdot\rangle_{S}$,
	$x_i = \sum_ j z_{ij} S^{-2} (\gamma_{\lambda, j}),$\\
	with  $\langle f,g\rangle_{S} = \langle S(f), S(g)\rangle = \langle S^2 (f), S^2(g)\rangle_\lambda.$
\end{enumerate}
\label{prop1} 
\end{Proposition}
Therefore, the eigenfunctions of the covariance operator $\mathcal{C}_{Sx}=S\mathcal{C}_{x} S$ of the smoothed sample $S(x_{i})$ are given by  $\beta_j = S^{-1}(\gamma_{\lambda,j}),$ where $\gamma_{\lambda,j}$ are obtained by the penalized estimation procedure set out for (\ref{eqsv}). 
Then, the basis $\beta_j$ is orthonormal with respect to the usual inner product in $L^{2}(T)$, so that the smooth sample data $S(x_{i})$ can be approximated by its truncated K-L expansion
\begin{equation}\label{FK-L}
	\x_{i}^{q}(t)=\sum_{j=1}^{q} z_{ij} \beta_{j},
\end{equation}
where $z_{ij} = \langle \beta_{j},S(x_{i})\rangle = \langle \gamma_{\lambda,j}, x_{i} \rangle,$ and $\x^q_i(t)$ denotes a $q$-dimensional orthonormal representation of the smoothed sample data $S(x_{i})$ in $L^{2}(T).$ The functional ICA version proposed in this paper uses the elements of this expansion to estimate the independent components of the original data. 

Our main assumption facts that the target functions can be found in the space spanned by the first $q$ eigenfunctions of the operator $\mathcal{C}_{Sx}$, as it is endowed with a smooth second-order structure represented by the major modes of variation of the empirical data. Thus, in such eigensubspace, it is expected to gain some accuracy in the forthcoming results due to the attenuation of the higher oscillation modes corresponding to the small eigenvalues of $\mathcal{C}_{Sx}$. Henceforth, we denote by $\mathcal{M}^{q}=\mathrm{span}\{ \beta_{1},\ldots,\beta_{q}\} $ the subspace spanned by the $q$ first eigenfunctions of $\mathcal{C}_{Sx}$. Without loss of generality, $\mathcal{M}^{q}$ will be assumed to preserve the inner product in $L^2(T)$. 

Most of the multivariate ICA methods require the standardization of the observed data with the inverse square root of the covariance matrix in order to remove any linear dependencies and normalize the variance along its dimensions. In infinite-dimensional spaces, however, covariance operators are not invertible giving rise to an ill-posed problem. As long as our signal is represented in $\mathcal{M}^{q}$, no regularization is needed and under moderate conditions, the inverse of the covariance operator can be well defined. Since standardization is a particular case of whitening (or sphering), we can generalize the procedure in the form of a whitening operator $\Psi$ that transforms a function in $\mathcal{M}^{q}$ into a standarized function on the same space. This implies that $\Psi(\x^{q})=\tilde{\x}^{q}$ is a standardized functional sample whose covariance operator  $\mathcal{C}_{\tilde{\x}^{q}} $ naturally satisfies to be the identity inside the space.

As an extension of the multivariate case,  the sample kurtosis operator of the standardized data is usually defined as
\begin{equation}
\mathcal{K}_{\tilde{\x}^{q}}(h)(s)  = \frac{1}{n}\sum_{i=1}^{n}\langle \tilde{\x}_{i}^{q},\tilde{\x}_{i}^{q}\rangle \langle \tilde{\x}_{i}^{q},h\rangle \tilde{\x}_{i}^{q}(s)=\left\langle k(s,\cdot),h\right\rangle,
\label{FFobi}
\end{equation}
where $k(s,t)=n^{-1}\sum_{i=1}^{n}\Vert \tilde{\x}_{i}^{q}\Vert ^{2}\tilde{\x}_{i}^{q}(s)\tilde{\x}_{i}^{q}(t)$ denotes the kurtosis kernel function of $\tilde{\x}^{q},$ and $h$ the function in $\mathcal{M}^{q}$ to be transformed. In the remainder of this article, it is assumed that the kurtosis operator is positive-definite, Hermitian and equivariant (see \cite{Li16}). Again, by Mercer's theorem its kernel admits the eigendecomposition 
\begin{equation*}
k(s,t)=\sum_{l=1}^{q}\rho_{l}\psi_l(s) \psi_{l}(t),\
\end{equation*}
where $\left\{ \rho_{l},\psi_{l}\right\} _{l=1}^q$ is a positive sequence of eigenvalues and related eigenfuntions. With this, we can define the independent components of $\x_{i}^{q}$ as mutually independent  variables with maximum kurtosis
given by
\[
\zeta_{il, \tilde{\x}^q}=\langle \tilde{\x}_{i}^{q},\psi_{l}\rangle.
\]

Challenging questions arise on how the 
Karhunen-Loève Theorem might be applied in this context. Intuitively, we note that this procedure leads to the expansion $\tilde{\x}_{i}^{q}(t) =\sum_{l=1}^{q} \zeta_{il,\tilde{\x}^q} \psi_{l}(t)$ which can be approximated in terms of $r$ eigenfuntions $\psi_{l}$ of interest, e.g. those associated with the independent components with extreme kurtosis values. Under mild conditions, this problem was solved in \cite{Li16,Virta20} by choosing $r=q$. However, there are other possibilities, such as considering $r<q$ or $\{ \psi_{1},\ldots, \psi_q \}$ as a basis of projection for either $x,\x^{q}$ or $\tilde{\x}^q$, in view of the fact that it preserves the four-order structure of the standardized data.

\section{Basis expansion estimation using a P-spline penalty}
\label{sec3}

In order to estimate the independent components from noisy discrete observations in Equation (\ref{modelofuncional}), it will be assumed that the tajectories belong to a finite-dimensional space of $L^2(T)$ spanned by a set of B-spline basis functions $\{\phi_{1}(t), \ldots, \phi_{p}(t)\}.$ Then, each sample curve can be expanded as 
\begin{equation}
	\label{expan}
x_i (t) = \sum_{j=1}^{p}a_{ij}\phi_{j}(t),
\end{equation}  
or, in matrix form, $x = A \phi,$ where $A$ is a coefficient matrix $A= (a_{ij}) \in\mathbb{R}^{n\times p}$ and $\phi=(\phi_{1},\ldots,\phi_{p})^\T$, $x=(x_{1},\dots,x_{n})^\T$ denote vector-valued functions. The basis coefficients for each sample curve  can be found by least squares approximation minimizing 
the mean squared error 
\begin{equation*}
\label{sq}
\textsc{mse}\ensuremath{\left(a_{i} \mid x_{i}\right)}=\ensuremath{\left(x_{i}-\Phi_{i} a_{i}\right)^{\T}\left(x_{i}-\Phi_{i} a_{i}\right)},
\end{equation*}
where $\Phi_{i}=\{ \phi_{j}(t_{ik}) \} \in\mathbb{R}^{m_{i}\times p}$ and $a_{i}=(a_{i1},\ldots,a_{ip})^\T$. For general guidance on both definition knots  and  order of B-splines, we refer
the reader to \cite{Ramsay05}. Although in this paper a non-penalized least squares approximation is assumed, \cite{Aguilera13b} give  a detailed account of how to estimate the basis coefficients  using different roughness penalty approaches (continuous and discrete) in terms of B-splines.

The next step consists of smoothing the sample curves in terms of the smoothed principal components and associated weight functions $\beta_j$ in (\ref{FK-L}). To do so, we next derive the P-spline FPCA approach developed in \cite{Aguilera13} that incorporates a discrete penalty based on $d$-order differences of adjacent B-spline coefficients (P-spline penalty) in the orthonormality constraint. Let us consider the B-spline basis  expansion of the covariance eigenfunctions $\gamma (t)=\phi(t)^\T b$, with $b=(b_{1}, \ldots, b_{p})^\T$ being its vector of basis coefficients, and a discrete P-spline roughness penalty function defined by $\textsc{pen}_{d}(\gamma)=b^\T P_{d}b,$ where $P_{d}\in\mathbb{R}^{p\times p}$ is the penalty matrix $P_{d}=\Delta^\T_{d} \Delta_{d}$, with $\Delta_{d}$ being a matrix representation of the $d$-order difference operator $R.$ Throughout the paper, we assume two order differences defining the penalty function $b^\T P_{2}b=\left(b_{1}-2b_{2}+b_{3}\right)^{2}+\cdots+(b_{p-2}-2b_{p-1}+b_{p})^{2}.$
This way, the inner product in (\ref{eqip}) is given in terms of B-splines expansions as 
\begin{equation*}
\left\langle f,g\right\rangle {}_{\lambda} =  
\mathtt{f}^\T \mathcal{G} \mathtt{g} +\lambda \mathtt{f}^\T P_{2}  \mathtt{g},
\end{equation*}
with $f= \phi^\T \mathtt{f}$, $g= \phi^\T \mathtt{g}$, and $\mathcal{G}=(\langle \phi_{j},\phi_{j'}\rangle) ,(j,j'=1,\ldots,p)$. Then, the maximization problem in (\ref{eqsv}) is equivalent to solve the following matrix problem:
\begin{equation} \label{penalizedsamplevariance}
b_{\lambda,j} =\mathrm{arg max} \frac{b^\T\mathcal{G}\Sigma_A \mathcal{G} b}{b^\T\left(\mathcal{G}+\lambda P_{2} \right)b},
\end{equation}
subject to the constraint $b^\T\left(\mathcal{G}+\lambda P_{2} \right)b_{\lambda,k} = 0$ for all $k<j,$ where $\Sigma_A =n^{-1} A^\T A$ and $\lambda\geq0$ is the penalty parameter used to control the trade-off between maximizing the sample variance and the strength of the penalty.

 Because B-spline basis are non-orthonormal with respect to the usual $L^2$ geometry, we can apply Cholesky factorization of the form $LL^\T=\mathcal{G}+\lambda P_{2} $ in order to find a non-singular matrix that allows us to operate in terms of the B-spline geometrical structure induced into $\mathbb{R}^{q}$. Then, finding the weight coefficients corresponds to solve the eigenvalue problem
\begin{equation} \label{eq:2}
    L^{-1}\mathcal{G} \Sigma_A \mathcal{G} (L^{-1})^\T v_j=\eta_j v_j,
\end{equation}
where $v_j=L^\T b_{\lambda, j}$ and the coefficients of $\gamma_{\lambda, j}$ are $b_{\lambda, j} = (L^{-1})^\T v_j$.
Therefore, we have obtained a set of orthonormal functions with respect to the inner product $\langle \cdot,\cdot \rangle_{\lambda}.$ The $j$th smoothed principal component is then given by 
\begin{equation*}
z_j=A \mathcal{G} b_{\lambda, j} = A \mathcal{G} (L^{-1})^\T v_j.
\end{equation*}
Thus, the problem is reduced to the multivariate PCA of the matrix $A \mathcal{G} (L^{-1})^\T$ in $\mathbb{R}^{q}$  (see \cite{Aguilera13} for a detailed study). From the results in \cite{Ocana99,Ocana07} we  deduce in this paper  the expression of the smoothing operator $S$ that provides the equivalence between this multivariate PCA  and the functional PCA of the smoothed data $S(x_i)$ in $L^{2}(T).$  

\begin{Proposition} \label{Prop2}
Given the basis expansion (\ref{expan}) for a random sample $\{ x_i \}$ of curves in $L^{2}(T),$ the  PCA of the matrix $A \mathcal{G} (L^{-1})^\T$ with the usual inner product in $\mathbb{R}^p$ is equivalent to all FPCA in Proposition \ref{prop1} with  the operator $S^2$ defined as $S^2(f) = \phi(t)^\T (\mathcal{G}+\lambda P_{d})^{-1} \mathcal{G} \mathtt{f}$, with $f= \phi(t)^\T \mathtt{f}.$
\end{Proposition}

\begin{proof}
Define, for all $f= \phi(t)^\T \mathtt{f}, \ g= \phi(t)^\T \mathtt{g},$ the new inner product $\langle f,g \rangle_K = \mathtt{f}^\T K  \mathtt{g}$ where  $K= D^\T D,$ with $D = L^{-1} \mathcal{G}^\T.$ Proposition 2 in \cite{Ocana07} proved that the  PCA of matrix $A D^\T$ with the usual inner product in $\mathbb{R}^p$  is equivalent to FPCA of $x_i$ with respect to $\langle \cdot,\cdot\rangle_K.$ That is, 
$x_i = \sum_j z_{ij} f_j$ with $f_j = \phi^\T  D^{-1} v_j,$
with $v_j$ being the eigenvectors of the matrix $A D^\T.$ Then, from Proposition \ref{prop1} in this paper, we have that $\langle S^2(f),S^2(g)\rangle_\lambda = \langle f,g\rangle_K.$ If we suppose that there exists a matrix $C$ such that $S^2(f) = \phi^\T C \mathtt{f},$ then  
$\langle S^2(f),S^2(g)\rangle_\lambda = \mathtt{f}^\T C^\T (\mathcal{G}+\lambda P_{d}) C \mathtt{g} = \mathtt{f}^\T D^\T D \mathtt{g}.$ As a consequence, $C^\T L L ^\T C = D^\T D$, so that $L^\T C = RD$ with $R$ being an orthonormal matrix $(RR^\T = I_p).$ Therefore, $S^2 (f) = \phi^\T \{(L^{-1})^\T R D \} \mathtt{f}.$ On the other hand, from Proposition \ref{prop1}, we have that $\gamma_j = S^2 (f_j)$ which implies that $(L^{-1})^\T v_j = (L^{-1})^\T R D D^{-1} v_j.$ As a consequence we obtain that $R=I_p$ and $S^2 (f) =   \phi^\T \{(L^{-1})^\T D\} \mathtt{f} =  \phi^\T \{(\mathcal{G}+\lambda P_d)^{-1} \mathcal{G}\} \mathtt{f}.$ 
\end{proof}

As a result, the principal components (scores) of $S(x_i)$ are given by $Z =  A \mathcal{G} (L^{-1})^\T V\,$ where $V$ is the matrix whose columns are the eigenvectors $v_j$ verifying Equation (\ref{eq:2}), and thus the eigenfunctions $\beta_j$ are $\beta_j = S^{-1} (\gamma_{\lambda,j}).$

Having estimated the weight functions coefficients and principal components scores, assume next that the smooth principal component expansion in (\ref{FK-L}) is truncated at the \emph{q}-term. Then, the column vector of smoothed sample curves is given by  $\x^{q}\left(t\right)= Z^q \beta \left(t\right),$ where  $Z^q=( z_{ij}) \in\mathbb{R}^{n\times q}$ is the matrix whose columns are the first  $q$ principal components   scores with respect to the basis of smooth principal component weight functions $\beta (t) = (\beta_1(t), \dots, \beta_q(t))^\T.$ 

With the above results, the functional independent components are computed from the smoothed principal component approximation of functional data. 
Following the ICA pre-processing steps, we first standardize the approximated curves defining the whitening operator as
$
    \Psi\{\x^{q}(t)\}=\tilde{\x}^{q}(t)= \tilde{Z}^q \beta (t),
$
with $\tilde{Z}^q = Z^q \Sigma_{Z^q}^{-\mathrm{1/2}}$ being the matrix of standardized principal components and 
$\Sigma_{Z^q }^{-1/2} = \sqrt{n} \{(Z^q)^\T Z^q\}^{-1/2},$ the inverse square root of the covariance matrix of $Z^q$. The described whitening transformation is essentially an orthogonalization of the probabilistic part of $\x^q$, so the matrix $\tilde{Z}^q \in\mathbb{R}^{n\times q}$ naturally satisfy $\Sigma_{\tilde{Z}^q}=I_{q}$, and the associated covariance operator $\mathcal{C}_{\tilde{x}^{q}}$ is unitary.

Then, the  kurtosis operator (\ref{FFobi}) of the standardized curves $\tilde{\x}^{q}(t)$ is given in matrix form by
\begin{equation*}
 \mathcal{K}_{\widetilde{\x}^{q}}(h) = \frac{1}{n} (\widetilde{Z}^{q^\T} D_{\widetilde{Z}^q} \widetilde{Z}^q \mathtt{h} )^\T \beta (t), \ \ \forall h = \beta(t)^\T \mathtt{h},
\end{equation*}
where $D_{\tilde{Z}^q} = \mathrm{diag}(\widetilde{Z}^q \widetilde{Z}^{q^\T}).$
The eigenanalysis of this kurtosis operator leads to the  diagonalization of the kurtosis matrix of the standardized principal components $\widetilde{Z}^q$, 
\begin{equation} \label{eq:3}
	\Sigma_{4,\widetilde{Z}^q} u_{l}= \rho_{l}u_{l} \quad (l=1,\ldots,q),
\end{equation}
where $\Sigma_{4,\widetilde{Z}^q}\in\mathbb{R}^{q\times q}$ is defined as
\[\Sigma_{4,\widetilde{Z}^q} = \frac{1}{n}\sum_{i=1}^{n}\left\Vert \tilde{z}_{i}^q\right\Vert ^{2}\tilde{z}_{i}^q\tilde{z}^{q^\T}_{i} =   \frac{1}{n}\widetilde{Z}^{q^\T} D_{\widetilde{Z}^q} \widetilde{Z}^q,
\]
with $\tilde{z}_i^q$ being the column vector $q\times 1$ with the $i$th row of the matrix $\tilde{Z}^q.$ The eigenproblem (\ref{eq:3}) is not restricting to assume that $\Sigma_{4,\widetilde{Z}^q} $ is uniquely determined. In fact, other kurtosis matrices can be considered (see, e.g., \cite{Kollo08,Loperfido17}). This way, the P-spline smoothed functional ICA of $x$ in $L^2(T)$ is obtained from the multivariate ICA of $Z^q$ in $\mathbb{R}^{q}.$ The resulting weight functions are now $\psi_{l}(t)=\beta(t)^\T  u_{l}\ (l=1,\ldots,q),$ where the coefficients vectors $u_l$ are the eigenvectors of the predefined kurtosis matrix. Then, the independent components can be calculated as $\zeta_{l,\tilde{\x}^q}=\tilde{Z}^q u_l.$  Finally, the  operator $\Gamma$ defining the FICA model is
\[\Gamma(\x_i^q) =  \beta^\T U^\T \Sigma_{Z^q}^{-1/2} z_i^q,
\]
 with $z_i^q$ being the column vector $q\times 1$ with the $i$th  row of $Z^q$ and $U\in \mathbb{R}^{q \times q}$ the matrix of  eigenvectors of the kurtosis matrix $\Sigma_{4,\widetilde{Z}^q}.$

\section{Parameter tuning}
\label{sec4}
The problem concerning the estimation of the smoothed independent component curves lies in finding an optimal truncation point $q$, as well as a suitable penalty parameter. As $q$ approaches $p$, more of the higher oscillation modes of the standardized sample are induced in the estimation. Otherwise, we are denoising the data from its second and fourth-order structure simultaneously. From this perspective, it is desirable to increase the value of $q$ such that the latent functions of the whitened space can be captured by the kurtosis operator. 
Observe that this kind of regularization is not exactly the same as the one providing the P-spline penalization of the roughness of the weight functions. Attenuating the higher frequency components of the FPCA model does not necessarily affect an entire frequency bandwidth of the data. Thus, if the original curves are observed with independent error, and the error is persistent in the functional approximation, it may overlap the estimation of the kurtosis eigenfunctions. In this context,  smoothing would be appropriate. Once the value of $q$ is decided, we should examine those components with extreme kurtosis, contrary to the FPCA where only the components associated to large eigenvalues are considered.

\subsection{Penalty parameter selection}
\label{Penalty}
Leave-one-out cross-validation \citep{Rice91} is generally used to select the penalty parameter in order to achieve a suitable degree of smoothness on the weight functions, but also to induce the truncation point $q$. In a more explicit and condensed form, this procedure in our model lies in finding a value of $\lambda$ that minimizes 
\begin{equation} \label{CV}
   \textsc{cv}_{q}(\lambda)=\frac{1}{n}\sum_{i=1}^{n}\left\Vert x_{i}-\x_{i}^{q(-i)}\right\Vert ^{2},
\end{equation}
where $\x_{i}^{q(-i)}=\sum_{l=1}^{q}z_{il}^{(-i)}\beta_{l}^{(-i)}(t)$ is the reconstruction of the $i$th curve  $x_i$ in terms of the $q$ first smoothed principal components by leaving out it in the estimation process. We found, however, that cross-validation was not sensitive for a reasonably large basis dimension, forcing us to reformulate the strategy. To address this problem, the penalty parameter might be subjectively chosen although this can lead to the bias and poor extraction of the artifactual sources. Hence, for the results presented in this paper, we propose a novel adaptive approach which consists in replacing (\ref{CV}) by
\begin{equation} \label{baselineCV}
   \textsc{bcv}_{q}(\lambda)=\frac{1}{n}\sum_{i=1}^{n}\left\Vert \x_{i}^{q; \lambda(-i)}-\x_{i}^{q;\lambda+\ell(-i)}\right\Vert ^{2},
\end{equation}
where $\x_{i}^{q;\lambda(-i)}$ is a smoothed representation of $x_i$ for some $\lambda$ and $\ell>0$ a value that increases the penalty in the second term of the norm, assume $\ell=0.1$. Then, for a fixed $q$,  (\ref{baselineCV}) is iterated for each $\lambda$ in  a given grid to find the one that minimizes $\textsc{bcv}_{q}(\lambda)$.  Among all the $q$ considered in the estimation process, we select the truncation point that minimizes this function.

If we require a basis dimension $p$ greater than sample size $n,$ a shrinkage covariance estimator \citep{Shafer05} can be considered for computing $\Sigma_A$. This method guarantees positive definiteness and consequently an estimation of the higher and important eigenvalues not biased upwards. The same strategy is used for $\textsc{bcv}_{q}(\lambda)$. Recall the quadratic distances in (\ref{baselineCV}). These are given in terms of basis functions by

{\small
\begin{gather*}
\left\Vert \x_{i}^{q;\lambda(-i)}-\x_{i}^{q;\lambda+\ell(-i)}\right\Vert ^{2}=\int_{T}\left[\x_{i}^{q;\lambda(-i)}(t)-\x_{i}^{q;\lambda+\ell(-i)}(t)\right]^{2}\text{d}t=\\
=\int_{T}\left[\sum_{l=1}^{q}z_{il}^{\lambda(-i)}\sum_{j=1}^{p}\mathtt{b}_{lj}^{\lambda(-i)}\phi_{j}(t)-\sum_{l=1}^{q}z_{il}^{\lambda+\ell(-i)}\sum_{j=1}^{p}\mathtt{b}_{lj}^{\lambda+\ell(-i)}\phi_{j}(t)\right]^{2}\text{d}t=\\=\int_{T}\left[\sum_{j=1}^{p}e_{ij}\phi_{j}(t)\right]^{2}\text{d}t=e^\T_{i}\mathcal{G} e_{i},
\end{gather*}
}where  $\mathtt{b}_{j}=(\mathtt{b}_{j1}, \ldots, \mathtt{b}_{jp})^{\T}$ is the vector of basis coefficients of the $j$th weight function $\beta_j$ in the B-spline basis $\phi_j(t)$ and $e_{i}=(e_{i1},\ldots,e_{ip})^\T$ is a vector of residuals. Next, the matrix $\mathcal{E}=(e_{ij})\in\mathbb{R}^{n\times q}$ is reconstructed via shrinkage. That is, first we compute cov$_{S}(\mathcal{E})$ where cov$_{S}$ is a predefined shrinkage covariance estimator, then we apply Cholesky decomposition of the form $LL^\T=\operatorname{cov}_{S}(\mathcal{E})$. Finally, the basis coefficients of the reconstructed residual functions are $\hat{e}_{i}=(L^{-1})^\T e_{i}$, and consequently now
\begin{equation*} 
\textsc{bcv}(\lambda)_{q}=\frac{1}{n}\sum_{i=1}^{n}\left\Vert \x_{i}^{q;\lambda(-i)}-\x_{i}^{q;\lambda+\ell(-i)}\right\Vert ^{2}=\hat{e}^\T_{i}\mathcal{G}\hat{e}_{i}.
\end{equation*}

We call this method \emph{baseline cross-validation}, as it operates across different reconstructions of $x_i$ for a given baseline penalty parameter and a fixed $q$. This approach is more versatile and particularly useful when the original curves are extremely rough and approximated with a larger basis dimension, thus avoiding the least squares to collapse. Moreover, for a given $q$, it allows to score more than one $\lambda$ as a result of the various relative minima it produces. The intuition behind baseline cross-validation is that there are several smoothing levels to endow the estimator with the ability for predictive modelling. These are given at evaluating "short distances" for a smoothing baseline $\lambda$ in a given $\x^q_i$, which may be seen as a way of finding a trade-off for the global roughness of a $q$-dimensional basis. Note that, as the value of $q$ increases, and despite the minimization of the mean squared error, it may be more difficult to find a smoothing balance between the elements of the basis  due to a complex fabric of variability modes.

\begin{algorithm}[H]
\scriptsize
\floatname{algorithm}{Algorithm}
\renewcommand{\thealgorithm}{}
\caption{\textsc{baseline cross-validation}}
\label{protocol1}
\textbf{Input}: $A, \phi_j\ (j=1,\ldots,p), \mathcal{G}, P_{2}, \lambda_k=(\lambda_1,\ldots,\lambda_m)^\T$\\
\textbf{Output}: $\lambda^\bullet.$\\
\textbf{for each}  $\lambda$ in $\lambda_k$:
\begin{algorithmic}[1]
\STATE Calculate $L^{-1}$ via Cholesky decomposition of the matrix $\mathcal{G}+\lambda P_{2}=L L^\T$ and for $\mathcal{G}+(\lambda+\ell) P_{2}=L L^\T$.
\STATE Diagonalize $L^{-1}\mathcal{G} \Sigma_{A_s}\mathcal{G} (L^{-1})^\T,$ where $\Sigma_{A_s}=\ $cov$_S(A)$, to obtain the coefficients of the eigenfunctions $\beta_j$, $\mathtt{b}_{j}$ and $\mathtt{b}_{\ell,j}$ for the incremental smoothing case .
\STATE Calculate $Z^q=A^{\T}\mathcal{G} b_{j}$, $Z^q_{\ell}=A^{\T}\mathcal{G}\mathtt{b}_{\ell,j}$ and $\mathcal{A}=\mathtt{b}_{j}(Z^q)^{\T}$, $\mathcal{A}_\ell=\mathtt{b}_{\ell,j}(Z^q_\ell)^{\T}$, where $\mathcal{A},\mathcal{A}_\ell$ are the coefficient matrices of the smoothed principal component expansion in terms of $\phi_j$.
\STATE $\mathcal{E}=\mathcal{A}-\mathcal{A}_\ell$ and reconstruct $\mathcal{E}$ via the covariance matrix cov$_{S}(\mathcal{E}).$\STATE  \textsc{bcv}$(\lambda)=n^{-1}$tr$(\hat{\mathcal{E}}^{\T}\mathcal{G}\hat{\mathcal{E}})$, where $\hat{\mathcal{E}}$ is the reconstructed matrix of residual coefficients and tr($\cdot$) is an operator that sums the diagonal elements of a square matrix.
\end{algorithmic}
\textbf{end for} \\
$\lambda^\bullet\leftarrow$ argmin$_{\lambda}$\textsc{bcv}.
\end{algorithm} 

 \section{Simulation study}
\label{sec5}
A simulation study based on  EEG data segments containing stereotyped artifacts was conducted to validate our methods for recovering brain sources. The data consists of 4 separate 64-channel recordings of a subject performing the following classes of self-paced repetitive movements: nodding, hand-tapping with a wide arm movement, eye-blinking and chewing. Recordings were performed in absence of sensory stimulation in a trial length 3 seconds sampled at 1 kHz, i.e., $t_{ik}\ (i=1,\ldots,64; k=1,\dots,3000)$. More details on the preprocessing steps and experimental conditions are deferred to the online supplementary material. In reconstructing the functional form of the sample paths, we sought a less smooth fitting to mimic the brain potential fluctuations. Accordingly, a basis of cubic B-spline functions of dimension $p=230$ is fitted to all signal components minimizing the mean squared error to a negligible value. 

The process of identifying artifactual functions is addressed by using topographic maps that roughly represent patterns of eigenactivity related to the distribution of bioelectric energy on the scalp. These maps are elaborated from the projection of the signal components $x_1, \ldots, x_{64}$ on to the basis of independent weight functions, i.e. $\zeta_{il,x}=\langle x_{i},\psi_{l}\rangle\ (i=1,\ldots,64; l=1,\ldots,q)$, whose resulting score vectors $\zeta_{l,\x}=\left(\zeta_{1l},\ldots,\zeta_{nl}\right)^\T$ are depicted in the spatial electrode domain. Therefore, the aim is to examine how the kurtosis eigenfunctions contribute into $x_i$ to discern possible patterns of artifactual activity. The components identified as artifacts will be considered for subtraction.

In order to simplify the burden of a manual selection, assume that all $\psi_1,\ldots,\psi_q$ obtained from the model correspond to a structure of latent artifactual eigenpatterns. Moreover, let $\x_{i}^{q}(t)=\sum_{l=1}^{q} \zeta_{il,x}\psi_{l}(t)$ be an expansion of artifactual components and related artifactual eigenfunctions.  Then, the artifact subtraction in terms of basis expansions is

\begin{equation} \label{subtractionIC}
\begin{split}
   x_{i}(t)-\x_{i}^{q}(t) &=\sum_{j=1}^{p}a_{ij}\phi_{j}(t)-\sum_{l=1}^{q}\zeta_{il,x}\sum_{j=1}^{p} (u^\T_l \mathtt{b}_{j}) \phi_{j}(t)=\\
   &=\sum_{j=1}^{p}d_{ij}\phi_{j}(t),
   \end{split}
\end{equation}
where $d_{ij}$ are the cleaned (or residual) coefficients, with $u_l$ being the vector of coefficients of the independent weight function $\psi_l$ in terms of the principal eigenfunctions. 
Thus, given the model parameters $q$ and $\lambda$, the procedure to estimate and remove smooth artifactual components from EEG functional data can succinctly be derived  as follows:

\begin{algorithm}[H]
\scriptsize
\floatname{algorithm}{Algorithm}
\renewcommand{\thealgorithm}{}
\caption{\textsc{functional artifact subtraction}}
\label{protocol1}
\textbf{Input}: $A, \phi_j\ (j=1,\ldots,p), \mathcal{G}, P_{2}, \lambda, q$\\
\textbf{Output}: $d_j.$
\begin{algorithmic}[1]
\STATE Calculate $L^{-1}$ via  Cholesky decomposition of the matrix $\mathcal{G}+\lambda P_{2}=L L^{\mathrm{T}}$.
\STATE Perform the PCA of $A\mathcal{G} (L^{-1})^\T$. Obtain $Z^q$ and the coefficients $\mathtt{b}_j$ of $\beta_j$. \\
$\rightarrow\qquad$ if  $p>n$ then diagonalize  $L^{-1}\mathcal{G} \Sigma_{A_s}\mathcal{G} (L^{-1})^\T,$ where $\Sigma_{A_s}=\ $cov$_S(A).$
\STATE Whiten $Z^q$: i.e. $\tilde{Z}^q=Z^q \Sigma_{Z^q}^{-1 / 2}$.
\STATE Fix a fourth-order matrix $\Sigma_{4,\widetilde{Z}^q}$ and diagonalize it. Obtain the eigenvalues $\rho_{l}$ and associated eigenvectors $u_{l}\ (l=1,\ldots,q)$.
\STATE  Calculate  $\zeta_{i l, x}=\left\langle x_{i}, \psi_{l}\right\rangle$ for $\psi_{l}(t)=\sum_{j=1}^{q} u_{l j} \beta_{j}(t)$.
\STATE Select the artifactual score vectors in $\zeta_{l, x}.$ Expand the artifactual space as $\x_{i}^{q}(t)=\sum_{l=1}^{q} \zeta_{il,x}\psi_{l}(t)$.
\STATE Subtract the artifactual coefficients in terms of $\phi_j$ using (\ref{subtractionIC}) and obtain the vector of coefficients $d_j$ to reconstruct the functional brain signal.
\end{algorithmic}
\end{algorithm} 
Baseline cross-validation was performed on a given grid, selecting the value which minimizes $\textsc{bcv}_q(\lambda)$ for $q=1,\dots,j_0$ where $j_0$ is defined as  the index entry corresponding to the first relative maximum of the    first order differences of FPCA's eigenvalues  $\Delta\eta_{j}.$ We find that truncating at $q=j_0$ is a way of exploring independence in the high variability structure of the data. In analysing EEG signals, this entails major effectiveness at reducing the artifactual content to a few eigenfunctions, particularly for the low-frequency physiological activity such as blinks and movement-related artifacts. One may see this truncation rule as a measure to improve the accuracy in the estimation of certain artifacts, while preserving the modes of variability related to the rhythms of the latent brain processes. The log-distances using \textsc{bcv}$(\lambda)$ for each one of the datasets are shown in Figure \ref{figl}. Further results are presented in Table \ref{Table1}.

\begin{figure}[h!]
\centering\includegraphics[scale=.52]{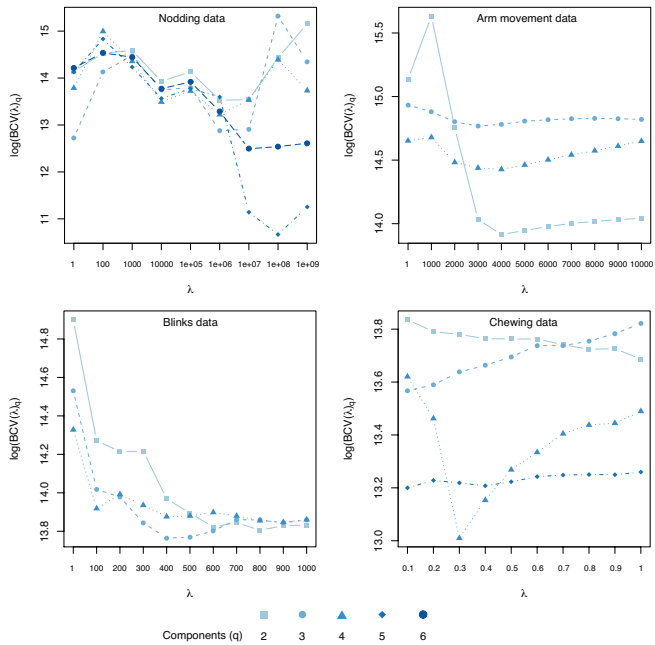}
\caption{The estimated log-\textsc{bcv}$(\lambda)$ function for the first components of each EEG dataset containing different classes of artifacts. }
\label{figl}
\end{figure}

\begin{table}[h!]
\caption{\label{Table1}Summary of parameters and cumulative variance  of the FICA model.}
\begin{tabular}{lccccccc}
\toprule
	& $j_0$ & $q$ & $\lambda$ & log-\textsc{bcv}$(\lambda)$ & var (\%) & var (\%) \\ 
Trial &&&&&$\lambda$&$\lambda=0$
 \\
\midrule
Nodding & 6 & 5 & $10^8$ & 10.66 & 99.40 & 94.43 \\ 
Arm mov. & 4 & 2 & 4000 & 13.91 & 75.85 & 62.42 \\  
Blinks & 4 & 3 & 400.0 & 13.76 & 97.50 & 93.56\\ 
Chewing & $5$ & $4$ & 0.300 & 13.01 & 68.23 & 68.03 \\
\botrule
\end{tabular}
\end{table}

\begin{figure*}[h!]
\centering\includegraphics[scale=.45]{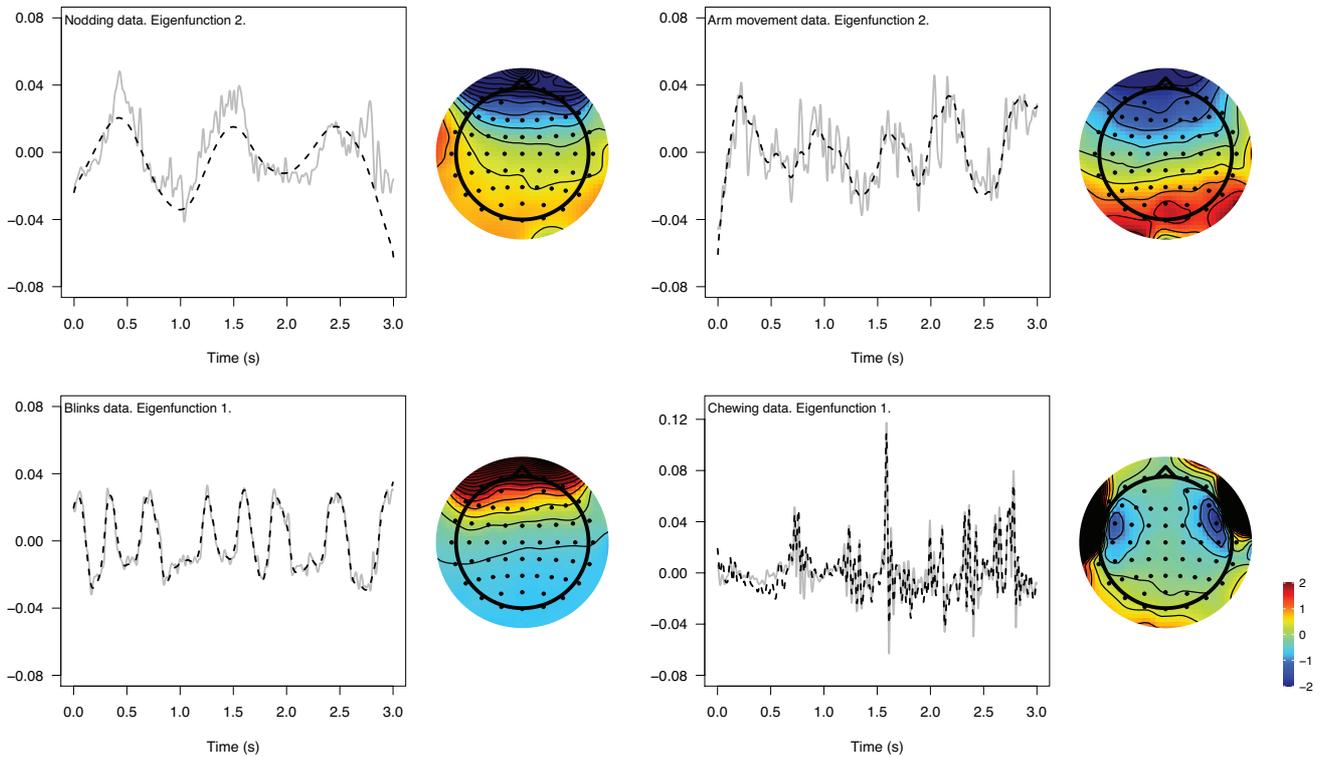}
\caption{Artifactual eigenfunctions selected from each trial. The unpenalized FICA (grey) and P-spline smoothed FICA (black dashed) decompositions are compared.  The scalp maps represent the scores depicted in the spatial electrode domain obtained by projection of the smooth eigenfunctions in the original sample.}
\label{Fig1}
\end{figure*}

\begin{figure*}[h!]
\centering\includegraphics[scale=.45]{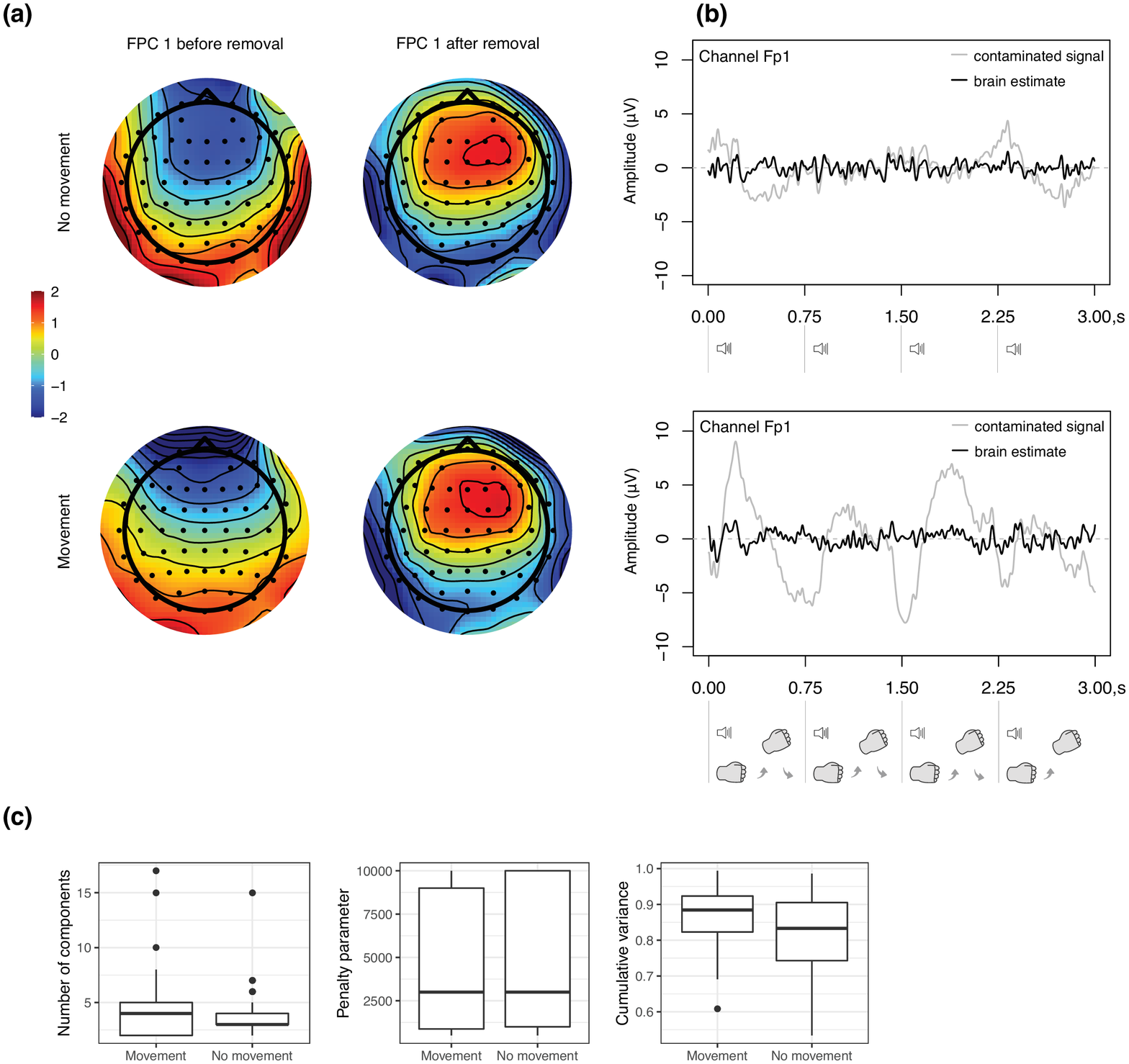}
\caption{(a) Topographic maps representing the leading functional principal component of the averaged trials before performing the P-spline smoothed FICA and after. (b) Grand-average across trials of a prefrontal channel where the artifactual activity is revealed. A descriptive scheme of the  movement is provided at the bottom of the plots. (c) Box plots of the number of components, the selected penality parameter and the cummulative variance of the model.}
\label{Fig3}
\end{figure*}

Preliminary results comparing both penalized and non-penalized estimation show that the smoothed FICA presumably attenuates the high-frequency potentials of neural origin, revealing the latent shape of the artifact. More importantly, however, is that all topographic maps reflect well-known spatial activation of the artifactual content. A selection of eigenfunctions from each trial and their associated component scores are depicted in Figure  \ref{Fig1}. Physiological non-brain activity near the recording zone, such as blinks and large amplitude body movements, can be easily detected in controlled conditions using the proposed methodology. However, the coexistence of such artifacts may result in a non-linear distortion of them, e.g., via large changes of the impedance \cite{Noury16}. This could entail a more challenging situation, as algorithms based on linear mixing may not be that effective at a certain point. Nonetheless, the aforementioned artifacts enhance the role of smoothing due to their low-frequency trademark in the signal. In contrast, when artifacts are characterised by localised high-amplitude curves, as it is the case of the fourth artifactual eigenfunction (chewing), smoothing is not able to denoise  effectively. We believe this happens for two reasons: first, the noise provided by the fourth-order structure of the model is essential to configure the shape of the artifact; second, the B-spline basis has a limited flexibility to smooth abrupt local contours. Hence, artifacts such as jaw clenching and chewing  are quite sensible to smoothing and difficult to correct for subtraction. Interestingly, hybrid procedures combining spline interpolation and wavelet filtering have shown promising results trying to solve this problem in functional near-infrared spectroscopy research (see \cite{Novi20}).

It seems reasonable to conjecture that restricting $q$ to the first FPCA terms decreases the odds of getting spurious artifactual functions, as they represent dominant modes of variability usually related to large artifacts. In such cases, the artifact subtraction with the smoothed components preserved  the brain activity rhythms in the original form, while for $\lambda=0$ it caused a reduction and a distortion of relevant potentials. However, \textsc{bcv}  may tend to oversmooth slightly in a sense of an effective artifact removal, resulting in certain artifactual residue after subtraction. This happens due to the  complexity of the mixed sources and can be solved by examining other relative minima in our results. The plots for all channels and datasets comparing the effect of subtracting artifactual components are omitted for the sake of space. Online supplementary materials provide \texttt{R} code for its visualization.

Although our tests have provided good results by subtracting all smoothed components, further research is needed to corroborate their physiological validity. As reported in \cite{Artoni18}, reducing the dimensionality of the data with a PCA before applying ICA is not always beneficial, although in some cases may improve the signal-to-noise ratio of the large sources and their subsequent isolation. We see that our approach paves the way for developing measures of correlation, dipolarity, stability or sparsity in the functional data domain to fine-tune artifact selection. An important issue that remains open is whether the restriction imposed for the truncation point is beneficial or not to achieve better results.

\section{Estimating brain signals from contaminated event-related potentials}
\label{sec6}
To illustrate our methods, we reproduced a typical experimental scenario where a human participant had to perform full-arm movements synchronised to a periodic auditory stimulus. An EEG recording was performed during the task. Arguably, what we provide here is a paradigmatic example wherein the researcher needs to clean the signal from motion-related artifacts while preserving activity genuinely related to perceptual and motor brain processes. The subject was instructed to tap his hand on the table synchronizing with a steady auditory stimulus in one condition while listening to the same stimulus without any movement involved in the other. Disposing of a baseline, we could directly compare the outcome of our cleaning procedure with an uncontaminated experimental situation. We recorded 100 trials of 3 seconds per condition, divided into randomized blocks of 25 trials. The stimulus period was $750$ milliseconds, i.e., 4 tappings in one of the conditions. Movements were intentionally exaggerated to maximize eventual movement-related artifacts.  In this section, the same configuration for running the model ($p=230; i=1,\ldots, 64; k=1,\ldots,3000 $) is preserved from the previous one.

The P-spline smoothed FICA is performed at each trial to obtain brain estimates by subtracting the artifactual components. Here, the complexity of the signal increases as it is assumed a mixture of artifacts and other brain processes due to the cognitive task.  Figure \ref{Fig3}  shows the grand-averaged results comparing both conditions before and after the artifact removal. A FPCA is performed on the averaged data to visualize the spatial distribution of the scores in the direction of the leading eigenvector  before and after the removal. As expected, the activations where nearly coincident after the artifact removal and more prominent in the central region of the scalp. The upper left panel displays the EEG signal in some frontal channels where the movement-related artifact is prominently visible before the subtraction. Further evidence of such artifactual content is given in  second row  where the raw curves are shown in the other condition. Clearly, the pooled artifacts across the trials have here a different origin. The same panel shows the curves after subtracting the artifactual curves. 

Our procedure notably reduces the movement-related artifact and renders the signal more stationary. Indeed, differences are smaller in the non-movement condition but, in either case, our algorithm is capable of reducing artifactual content while retaining the brain activity intact. From our previous tests, one may expect some artifact residue at a trial level depending on the estimated $\lambda$ and the diversity of source artifacts. We stress that as the response to the repeated stimulus is assumed to be invariant and small in terms of amplitude, averaging suppresses non-phase-locked activity and reveals the potential elicited by the stimulus \citep{Tong2009QuantitativeEA}. Consequently,  the attenuation of the roughness of the artifactual component functions will lead to a better estimation of brain potentials at averaging rather than the subtraction of rough components.

\section{Discussion}
\label{sec7}

The proposed independent techniques are, to the best of our knowledge, the first to provide a functional framework for smoothed artifact extraction and removal of dense data approximated with a large number of knots. We found that using shrinkage estimators is a reasonable starting point for smoothing covariance operators with this kind of functional data (see also \cite{Ieva2016}). According to this setting, a novel cross-validation method has been proposed for selecting the model parameters. Despite being computationally expensive, our approach has proven to outperform the lack of sensitivity of other existent methods. Overall, this allows the application of independent component techniques from a smoothing perspective somewhat more flexible when compared to other modelling strategies.

Although \cite{Li16} established a form of Fisher consistency for the kurtosis operator decomposition, no asymptotic results of the non-smoothed and, hence of the smoothed independent components have been derived. Therefore, one can assume that we rely on a competitive performance derived from previous FPCA asymptotic results. In our empirical setting, however, the study of such properties must be related to the functional data type and the penalized spline method used, involving considerably more technicalities. See, for example, \cite{Zhang16} and \cite{Xiao19}. These theoretical developments lie beyond the scope of the present work. However, we hope to pursue such study in a separate paper. 

In our simulations, the kurtosis operator has proven to work well at capturing artifactual eigenfunctions with different frequency characteristics, at least under certain conditions. One of the strengths of our model is the double regularization, which allows us to circumvent the leak of brain activity and get clean movement-related artifacts. In essence, the degree of separation is defined through the space dimension, from more dependent (first $q$ terms of the FPCA decomposition) to more independent ($q \rightarrow p$). Thus $q$ acts as a regularization parameter to explore the variational component of the artifactual sources in the EEG signal, while $\lambda$ provides more accurate estimations, particularly in using the first $q$ terms of the K-L expansion. Further research is needed to determine how the model parameter selection can optimize the removal of artifacts with a minimum loss of variance patterns related to brain sources. Non-linear artifact distortion will inevitably suffer from cortical entrainment of challenging correction, suggesting the exploration of other subspaces prone to kurtosis data structures in addition to the smoothed principal component eigendirections.
\\
\\
\\
{\sf \bf{Author Contributions:}} 
Conceptualization, methodology, software and formal analysis, M.V. and A.M.A.; writing--original draft preparation, M.V.; writing--review and editing, M.V., M.R. and A.M.A; visualization, M.V.; supervision, A.M.A.; data collection and preprocessing M.V. and M.R. All authors have read and agreed to the published version of the manuscript.
\\
\\
{\sf \bf{Funding:}} 
This research is supported by the Methusalem funding from the Flemish Government and by the project MTM2017-88708-P of the Spanish Ministry of Science, Innovation and Universities, project FQM-307 of the Government of Andalusia (Spain).
\\
\\
{\sf \bf{Acknowledgments:}}
The authors wish to thank Marc Leman for its valuable comments and Daniel Gost for helping with figures and formatting.

\newpage

\bibliographystyle{apalike-dashed}

\bibliography{refs}


\end{document}